\documentclass[11pt,leqno]{article}

\usepackage{amsmath}
\usepackage{amsthm}
\usepackage{amssymb}

\theoremstyle{plain}
\newtheorem{theorem}{Theorem}[section]
\newtheorem{lemma}[theorem]{Lemma}
\newtheorem{proposition}[theorem]{Proposition}

\theoremstyle{definition}
\newtheorem{definition}[theorem]{Definition}

\theoremstyle{remark}
\newtheorem{remark}[theorem]{Remark}

\numberwithin{equation}{section}

\begin{document}

\title{\textbf{Is the solution to the BCS gap equation continuous in the temperature ?}}

\author{Shuji Watanabe\\
Division of Mathematical Sciences\\
Graduate School of Engineering, Gunma University\\
4-2 Aramaki-machi, Maebashi 371-8510, Japan\\
Email: shuwatanabe@gunma-u.ac.jp}

\date{}

\maketitle

\begin{abstract}
%%%%%%%%%%%%%%%%%%%%%%%%%%%%%%%%%%%
We regard the BCS gap equation in superconductivity as a nonlinear integral equation on a Banach space consisting of continuous functions of both $T$ and $x$. Here, $T (\geq 0)$ stands for the temperature and $x$ the kinetic energy of an electron minus the chemical potential. We show that the unique solution to the BCS gap equation is continuous with respect to both $T$ and $x$ when $T$ is less than or equal to a certain value. The proof is carried out based on the Banach fixed-point theorem.
%%%%%%%%%%%%%%%%%%%%%%%%%%%%%%%%%%%

\medskip


\medskip

\end{abstract}

%\runningtitle{Continuity of the solution to the BCS gap equation}
%\runningauthor{Shuji Watanabe}

%%%%%%%%%%%%%%%%%%%%%%%%%%% 1 %%%%%%%%%%%%%%%%%%%%
\section{Introduction and preliminaries}

We use the unit $k_B=1$, where $k_B$ stands for the Boltzmann constant. Let $\omega_D>0$ and $k\in\mathbb{R}^3$ stand for the Debye frequency and the wave vector of an electron, respectively. Let $m>0$ and $\mu>0$ stand for the electron mass and the chemical potential, respectively. We denote by $T (\geq 0)$ the temperature, and by $x$ the kinetic energy of an electron minus the chemical potential, i.e., $x=\hslash^2|k|^2/(2m)-\mu$. Note that $0<\hslash\omega_D<<\mu$.

In the BCS model (see \cite{bcs, bogoliubov}) of superconductivity, the solution to the BCS gap equation \eqref{eq:gapequation} below is called the gap function. We regard the gap function as a function of both $T$ and $x$, and denote it by $u$, i.e., $u: \, (T,\, x) \mapsto u(T,\, x)$ $(\geq 0)$. The BCS gap equation is the following nonlinear integral equation:
\begin{equation}\label{eq:gapequation}
u(T,\, x)=\int_{\varepsilon}^{\hslash\omega_D}
\frac{U(x,\,\xi)\, u(T,\,\xi)}{\,\sqrt{\,\xi^2+u(T,\,\xi)^2\,}\,}\,
\tanh \frac{\,\sqrt{\,\xi^2+u(T,\,\xi)^2\,}\,}{2T}\, d\xi,
\quad \varepsilon \leq x \leq \hslash\omega_D \, ,
\end{equation}
where $U(x,\,\xi)>0$ is the potential multiplied by the density of states per unit energy at the Fermi surface and is a function of $x$ and $\xi$. In \eqref{eq:gapequation} we introduce $\varepsilon>0$, which is small enough and fixed $(0<\varepsilon<<\hslash\omega_D)$. In the BCS model, the integration interval is $[0, \, \hslash\omega_D]$. However, we introduce very small $\varepsilon>0$ for the following mathematical reason. In order to show the continuity of the solution to the BCS gap equation with respect to the temperature, we make the form of the BCS gap equation somewhat easier to handle. That is to say, we let the integration interval be the closed interval $[\varepsilon, \, \hslash\omega_D]$ as in \eqref{eq:gapequation}.

It is known that the BCS gap equation \eqref{eq:gapequation} is based on a superconducting state called the BCS state. In this connection, see \cite[(6.1)]{watanabe} for a new gap equation based on a superconducting state having a lower energy than the BCS state.

The integral with respect to $\xi$ in \eqref{eq:gapequation} is sometimes replaced by the integral over $\mathbb{R}^3$ with respect to the wave vector $k$. Odeh \cite{odeh}, and Billard and Fano \cite{billardfano} established the existence and the uniqueness of the positive solution to the BCS gap equation in the case $T=0$. In the case $T\geq 0$, Vansevenant \cite{vansevesant} determined the transition temperature (the critical temperature) and showed that there is a unique positive solution to the BCS gap equation. Recently, Hainzl, Hamza, Seiringer and Solovej \cite{hhss} proved that the existence of a positive solution to the BCS gap equation is equivalent to the existence of a negative eigenvalue of a certain linear operator to show the existence of a transition temperature. Hainzl and Seiringer \cite{haizlseiringer} also derived upper and lower bounds on the transition temperature and the energy gap for the BCS gap equation. Moreover, Frank, Hainzl, Naboko and Seiringer \cite{fhns} gave a rigorous analysis of the asymptotic behavior of the transition temperature at weak coupling.

Since the existence and the uniqueness of the solution are established for fixed $T$ in previous papers, the temperature dependence of the solution is not covered. Studying the temperature dependence of the solution is very important. This is because, by dealing with the thermodynamical potential, this study leads to the mathematical challenge of showing that the transition to a superconducting state is a second-order phase transition. Indeed, when one tries to show that the transition to a superconducting state is a second-order phase transition, one has to differentiate the thermodynamical potential with respect to the temperature twice. Since the form of the thermodynamical potential includes the solution to the BCS gap equation, one has to differentiate the solution with respect to the temperature twice. So it is highly desirable to show the smoothness of the solution with respect to the temperature.

In this paper we address the continuity of the solution to the BCS gap equation \eqref{eq:gapequation} with respect to the temperature. We regard the BCS gap equation \eqref{eq:gapequation} as a nonlinear integral equation on a Banach space consisting of continuous functions of both $T$ and $x$. On the basis of the Banach fixed-point theorem, we show that the solution to the BCS gap equation \eqref{eq:gapequation} is continuous with respect to both $T$ and $x$ when $T$ is less than or equal to a certain value.

Let
\begin{equation}\label{eq:potentialu}
U(x,\,\xi)=U_1 \qquad \mbox{at all} \quad (x,\,\xi) \in [\varepsilon, \, \hslash\omega_D]^2,
\end{equation}
where $U_1>0$ is a constant. Then the gap function depends on the temperature $T$ only. In this case, we denote the gap function by $\Delta_1$, i.e., $\Delta_1:\, T \mapsto \Delta_1(T)$. Then \eqref{eq:gapequation} leads to the simple gap equation
\begin{equation}\label{eq:smplgapequation}
1=U_1\int_{\varepsilon}^{\hslash\omega_D}
 \frac{1}{\,\sqrt{\,\xi^2+\Delta_1(T)^2\,}\,}\,
 \tanh \frac{\, \sqrt{\,\xi^2+\Delta_1(T)^2\,}\,}{2T}\,d\xi.
\end{equation}

We now define the temperature $\tau_1>0$, which is the transition temperature originating from the simple gap equation \eqref{eq:smplgapequation}.
\begin{definition}[\cite{bcs}]\label{dfn:tcsmpl}
The transition temperature is the temperature $\tau_1>0$ satisfying
\[
1=U_1\int_{\varepsilon}^{\hslash\omega_D}
\frac{1}{\,\xi\,}\,\tanh \frac{\xi}{\,2\tau_1\,}\,d\xi.
\]
\end{definition}

\begin{remark}
There is another definition of the transition temperature, which originates from the BCS gap equation \eqref{eq:gapequation}. See \cite[Definition 2.5]{watanabe3}.
\end{remark}

The BCS model makes the assumption that there is a unique solution $\Delta_1: T \mapsto \Delta_1(T)$ to the simple gap equation \eqref{eq:smplgapequation} and that it is of class $C^2$ with respect to the temperature $T$ (see e.g. \cite{bcs} and \cite[(11.45), p.392]{ziman}). The author \cite{watanabe2} gave a mathematical proof of this assumption on the basis of the implicit function theorem. Set
\begin{equation}\label{eq:delta0}
\Delta=\frac{\,
\sqrt{ \left( \hslash\omega_D-\varepsilon\, e^{1/U_1} \right)
\left( \hslash\omega_D-\varepsilon\, e^{-1/U_1} \right) }\,}
{\,\sinh\frac{1}{\,U_1\,}\,}.
\end{equation}

\begin{proposition}[{\cite[Proposition 2.2]{watanabe2}}]\label{prp:simplegap}
Let $\Delta$ be as in \eqref{eq:delta0}. Then there is a unique nonnegative solution $\Delta_1: [\,0,\,\tau_1\,] \to [0,\,\infty)$ to the simple gap equation \eqref{eq:smplgapequation} such that the solution $\Delta_1$ is continuous and strictly decreasing on the closed interval $[\,0,\,\tau_1\,]$:
\[
\Delta_1(0)=\Delta>\Delta_1(T_1)>\Delta_1(T_2)>\Delta_1(\tau_1)=0, \qquad 0<T_1<T_2<\tau_1.
\]
Moreover, it is of class $C^2$ on the interval $[\,0,\,\tau_1\,)$ and satisfies
\[
\Delta_1'(0)=\Delta_1''(0)=0 \quad \mbox{and} \quad \lim_{T\uparrow \tau_1} \Delta_1'(T)=-\infty.
\]
\end{proposition}

\begin{remark}
We set $\Delta_1(T)=0$ for $T>\tau_1$.
\end{remark}

Let $0<U_1<U_2$ , where $U_2>0$ is a constant. We assume the following:
\begin{equation}\label{eq:condition}
U_1 \leq U(x,\,\xi) \leq U_2 \quad \mbox{at all} \quad (x,\,\xi) \in [\varepsilon,\, \hslash\omega_D]^2, \qquad U(\cdot,\,\cdot) \in C([\varepsilon,\, \hslash\omega_D]^2).
\end{equation}
When $U(x,\,\xi)=U_2$ at all $(x,\,\xi) \in [\varepsilon,\, \hslash\omega_D]^2$, an argument similar to that in Proposition \ref{prp:simplegap} gives that there is a unique nonnegative solution $\Delta_2: [\,0,\,\tau_2\,] \to [0,\,\infty)$ to the simple gap equation
\begin{equation}\label{eq:smplgapequation2}
1=U_2\int_{\varepsilon}^{\hslash\omega_D}
 \frac{1}{\,\sqrt{\,\xi^2+\Delta_2(T)^2\,}\,}\,
 \tanh \frac{\, \sqrt{\,\xi^2+\Delta_2(T)^2\,}\,}{2T}\,d\xi, \qquad
0\leq T\leq \tau_2.
\end{equation}
Here, $\tau_2>0$ is defined by
\[
1=U_2\int_{\varepsilon}^{\hslash\omega_D}
\frac{1}{\,\xi\,}\,\tanh \frac{\xi}{\,2\tau_2\,}\,d\xi.
\]
We again set $\Delta_2(T)=0$ for $T>\tau_2$.

\begin{lemma}[{\cite[Lemma 1.5]{watanabe3}}] \quad {\rm (a)} The inequality $\tau_1<\tau_2$ holds.

\noindent {\rm (b)} If \   $0\leq T<\tau_2$, then $\Delta_1(T)<\Delta_2(T)$. If \  $T\geq \tau_2$, then $\Delta_1(T)=\Delta_2(T)=0$.
\end{lemma}

Let $0 \leq T \leq \tau_2$ and fix $T$. The author considered the Banach space $C([\varepsilon,\, \hslash\omega_D])$ consisting of continuous functions of $x$ only, and dealt with the following subset $V_T$:
\begin{equation}\label{eq:vt}
V_T=\left\{ u(T,\,\cdot) \in C([\varepsilon,\, \hslash\omega_D]): \; \Delta_1(T) \leq u(T,\,x) \leq \Delta_2(T) \;\; \mbox{at} \;\; x \in [\varepsilon,\, \hslash\omega_D] \right\}.
\end{equation}

\begin{remark}
The set $V_T$ depends on $T$. So we denote each element of $V_T$ by $u(T,\,\cdot)$.
\end{remark}

In order to show how the solution varies with the temperature, the author \cite{watanabe3} gave another proof of the existence and the uniqueness of the solution to the BCS gap equation \eqref{eq:gapequation}; the author showed that the unique solution belongs to $V_T$:
\begin{theorem}[{\cite[Theorem 2.2]{watanabe3}}]\label{thm:solutionu0}
Let $T \in [0,\, \tau_2]$ be fixed. Then there is a unique nonnegative solution $u_0(T,\,\cdot) \in V_T$ to the BCS gap equation \eqref{eq:gapequation}:
\[
u_0(T,\, x)=\int_{\varepsilon}^{\hslash\omega_D}
\frac{U(x,\,\xi)\, u_0(T,\, \xi)}{\,\sqrt{\,\xi^2+u_0(T,\, \xi)^2\,}\,}\,
\tanh \frac{\,\sqrt{\,\xi^2+u_0(T,\, \xi)^2\,}\,}{2T}\, d\xi, \quad
x \in [\varepsilon,\, \hslash\omega_D].
\]
Consequently, the solution is continuous with respect to $x$ and varies with the temperature as follows:
\[
\Delta_1(T) \leq u_0(T,\, x) \leq \Delta_2(T) \quad \mbox{at} \quad
(T,\,x) \in [0,\, \tau_2] \times [\varepsilon,\, \hslash\omega_D].
\]
\end{theorem}

\begin{remark}
In Theorem \ref{thm:solutionu0} the author assumed the following:
\[
U_1 \leq U(x,\,\xi) \leq U_2 \quad \mbox{at all} \quad (x,\,\xi) \in [\varepsilon,\, \hslash\omega_D]^2, \qquad U(\cdot,\,\cdot) \in C^2([\varepsilon,\, \hslash\omega_D]^2).
\]
But Theorem \ref{thm:solutionu0} holds true under \eqref{eq:condition}.
\end{remark}

\begin{remark}
We regard the gap function of Theorem \ref{thm:solutionu0} as a function of both $T$ and $x$, and denote it by $u_0$, i.e., $u_0:\, (T,\,x) \mapsto u_0(T,\,x)$.
\end{remark}

Studying smoothness of the thermodynamical potential with respsect to $T$, the author \cite[Theorem 2.11]{watanabe3} showed, under a certain approximation, that the transition to a superconducting state is a second-order phase transition without the restriction \eqref{eq:potentialu} imposed in our recent paper \cite{watanabe2}.

The paper proceeds as follows. In section 2 we state our main results without proof. In section 3 we prove our main results.

%%%%%%%%%%%%%%%%%%%%%%%%%%%%%%%%  2  %%%%%%%%%%%%%%
\section{Main results}

Let $U_0 >0$ be a constant satisfying $U_0 < U_1 < U_2$. An argument similar to that in Proposition \ref{prp:simplegap} gives that there is a unique nonnegative solution $\Delta_0: [\,0,\,\tau_0\,] \to [0,\,\infty)$ to the simple gap equation
\[
1=U_0\int_{\varepsilon}^{\hslash\omega_D}
 \frac{1}{\,\sqrt{\,\xi^2+\Delta_0(T)^2\,}\,}\,
 \tanh \frac{\, \sqrt{\,\xi^2+\Delta_0(T)^2\,}\,}{2T}\,d\xi, \qquad
0\leq T\leq \tau_0.
\]
Here, $\tau_0>0$ is defined by
\[
1=U_0\int_{\varepsilon}^{\hslash\omega_D}
\frac{1}{\,\xi\,}\,\tanh \frac{\xi}{\,2\tau_0\,}\,d\xi.
\]
We set $\Delta_0(T)=0$ for $T>\tau_0$. A straightforward calculation gives the following.
\begin{lemma}\label{lm:taudelta} \quad {\rm (a)} \  $\tau_0<\tau_1<\tau_2$ .

\noindent {\rm (b)} If \   $0\leq T<\tau_0$, then $0<\Delta_0(T)<\Delta_1(T)<\Delta_2(T)$.

\noindent {\rm (c)} If \   $\tau_0 \leq T<\tau_1$, then $0=\Delta_0(T)<\Delta_1(T)<\Delta_2(T)$.

\noindent {\rm (d)} If \   $\tau_1 \leq T<\tau_2$, then $0=\Delta_0(T)=\Delta_1(T)<\Delta_2(T)$.

\noindent {\rm (e)} If \   $\tau_2 \leq T$, then $0=\Delta_0(T)=\Delta_1(T)=\Delta_2(T)$.
\end{lemma}

\begin{remark}
Let the functions $\Delta_k$ $(k=0,\, 1,\, 2)$ be as above. For each $\Delta_k$, there is the inverse $\Delta_k^{-1}: \, [0,\,\Delta_k(0)] \to [0,\, \tau_k]$. Here,
\[
\Delta_k(0)=\frac{\,
\sqrt{ \left( \hslash\omega_D-\varepsilon\, e^{1/U_k} \right)
\left( \hslash\omega_D-\varepsilon\, e^{-1/U_k} \right) }\,}
{\,\sinh\frac{1}{\,U_k\,}\,}
\]
and $\Delta_0(0)<\Delta_1(0)<\Delta_2(0)$. See \cite{watanabe2} for more details.
\end{remark}

Let $T_1$ satisfy $\displaystyle{ \, (0<) \, T_1 < \Delta_0^{-1}\left( \frac{ \,\Delta_0(0)\,}{2} \right) }$ and
\begin{equation}\label{eq:t1condition}
\frac{\,\Delta_0(0)\,}{\, 4\, \Delta_2^{-1}\left( \Delta_0(T_1) \right)\,} \tanh \frac{\,\Delta_0(0)\,}{\, 4\, \Delta_2^{-1}\left( \Delta_0(T_1) \right)\,} > \frac{1}{\, 2\,}\left( 1+\frac{\, 4\hslash^2\omega_D^2\,}{\,\Delta_0(0)^2\,} \right).
\end{equation}
Since $T<\Delta_2^{-1}\left( \Delta_0(T) \right)$, \  the temparature $T \in [0,\, T_1]$ satisfies
\[
\frac{\,\Delta_0(0)\,}{\, 4\, T\,} \tanh \frac{\,\Delta_0(0)\,}{\, 4\, T\,} > \frac{1}{\, 2\,}\left( 1+\frac{\, 4\hslash^2\omega_D^2\,}{\,\Delta_0(0)^2\,} \right).
\]

\begin{remark}
Numerically, $\hslash^2\omega_D^2/\Delta_0(0)^2=O(10^4)$ by experiments. Hence the temperature $T_1$ is very small as long as it satisfies \eqref{eq:t1condition}.
\end{remark}

\begin{remark}
We assume that $X > \Delta_0(0)/2$ in section 3. In accordance with this, we set $T_1 < \Delta_0^{-1}( \Delta_0(0)/2 )$ above.
\end{remark}

We now consider the Banach space $C([0,\, T_1] \times [\varepsilon,\,\hslash\omega_D])$ consisting of continuous functions of both $T$ and $x$, and deal with the following subset $V$ of the Banach space $C([0,\, T_1] \times [\varepsilon,\,\hslash\omega_D])$:

\begin{eqnarray}\label{eq:spacev}
V &=& \left\{ u \in C([0,\, T_1] \times [\varepsilon,\,\hslash\omega_D]): \Delta_1(T) \leq u(T,\,x) \leq \Delta_2(T) \right. \\
 & & \qquad \qquad \left. at \; \; (T,\,x) \in [0,\, T_1] \times [\varepsilon,\,\hslash\omega_D] \right\}. \nonumber
\end{eqnarray}

\begin{theorem}\label{thm:continuityu0}
Assume \eqref{eq:condition}. Let $u_0$ be as in Theorem \ref{thm:solutionu0} and $V$ as in \eqref{eq:spacev}. Then $u_0 \in V$. Consequently, the gap function $u_0$ is continuous on $[0,\, T_1] \times [\varepsilon,\,\hslash\omega_D]$.
\end{theorem}

%%%%%%%%%%%%%%%%%%%%%%%%%%%%%%%%  3  %%%%%%%%%%%%%%
\section{Proof of Theorem \ref{thm:continuityu0}}

Let $\xi \in [\varepsilon,\,\hslash\omega_D]$ and $X \in \left( \Delta_0(0)/2,\,\infty\right)$ be fixed. Then we can regard the following function $g$ given by
\begin{equation}\label{eq:functiong}
g(T;\, \xi,\, X)=\frac{1}{\,\left( \xi^2+X^2 \right)^{3/2} \,}\left\{ \xi^2\tanh Y+\frac{X^2\, Y}{\,\cosh^2 Y\,} \right\}, \quad Y=\frac{\,\sqrt{\xi^2+X^2}\,}{2T}
\end{equation}
as a function of $T$ $(\geq 0)$ only. Note that $g(T;\, \xi,\, X)>0$.

\begin{remark}
When $T=0$, \   $g(T;\, \xi,\, X)$ is regarded as $\displaystyle{ \frac{\xi^2}{\,\left( \xi^2+X^2 \right)^{3/2}\,} }$, i.e.,
\[
g(0;\, \xi,\, X)=\frac{\xi^2}{\,\left( \xi^2+X^2 \right)^{3/2}\,}.
\]
\end{remark}

Let $T_2>0$ satisfy
\begin{equation}\label{eq:tcondition}
\frac{\,\sqrt{\xi^2+X^2}\,}{2T_2}\tanh \frac{\,\sqrt{\xi^2+X^2}\,}{2T_2}>
\frac{1}{\, 2\,}\left( 1+\frac{\,\xi^2\,}{\,X^2\,} \right).
\end{equation}
\begin{lemma}\label{lm:increasing}\  Let $T_2$ be as in \eqref{eq:tcondition}. Then $g$ is continuous and strictly increasing on $[0,\, T_2]$.
\end{lemma}

\begin{proof}\quad At $T \in (0,\, T_2)$,
\[
\frac{\, \partial g\,}{\partial T}(T;\, \xi,\, X)=\frac{2Y^2}{\,\left( \xi^2+X^2\right)^2\cosh^2Y\,}\left\{ 2X^2Y\tanh Y-\left( \xi^2+X^2\right) \right\} >0.
\]
\end{proof}

Define a mapping $A$ by
\[
Au(T,\, x)=\int_{\varepsilon}^{\hslash\omega_D}
\frac{U(x,\,\xi)\, u(T,\, \xi)}{\,\sqrt{\,\xi^2+u(T,\, \xi)^2\,}\,}\,
\tanh \frac{\,\sqrt{\,\xi^2+u(T,\, \xi)^2\,}\,}{2T}\, d\xi,\qquad
u \in V.
\]

A straightforward calculation gives the following.
\begin{lemma} \quad Let $V$ be as in \eqref{eq:spacev}. Then $V$ is closed.\end{lemma}

\begin{lemma}\label{lm:auc} \quad $Au \in C([0,\, T_1] \times [\varepsilon,\,\hslash\omega_D])$ for $u \in V$.
\end{lemma}

\begin{lemma}\label{lm:d1aud2} \quad Let $u \in V$. Then \   $\Delta_1(T) \leq Au(T,\,x) \leq \Delta_2(T)$ at $(T,\,x) \in [0,\, T_1] \times [\varepsilon,\,\hslash\omega_D]$.
\end{lemma}

\begin{proof}\quad We show $Au(T,\, x) \leq \Delta_2(T)$.
Since
\[
\frac{u(T,\, \xi)}{\,\sqrt{\,\xi^2+u(T,\, \xi)^2\,}\,} \leq
\frac{\Delta_2(T)}{\,\sqrt{\,\xi^2+\Delta_2(T)^2\,}\,},
\]
it follows from \eqref{eq:smplgapequation2} that
\begin{eqnarray*}
Au(T,\, x) \leq \int_{\varepsilon}^{\hslash\omega_D}
\frac{U_2\, \Delta_2(T)}{\,\sqrt{\,\xi^2+\Delta_2(T)^2\,}\,}\,
\tanh \frac{\,\sqrt{\,\xi^2+\Delta_2(T)^2\,}\,}{2T}\, d\xi
=\Delta_2(T).
\end{eqnarray*}
The rest can be shown similarly by \eqref{eq:smplgapequation}.
\end{proof}

Combining Lemma \ref{lm:d1aud2} with Lemma \ref{lm:auc} immediately yields the following.
\begin{lemma}
Let $u \in V$. Then \   $Au \in V$.
\end{lemma}

We now show that the mapping $A: \, V \longrightarrow V$ is contractive. We denote by $\left\| \cdot \right\|$ the norm of the Banach space $C([0,\, T_1] \times [\varepsilon,\,\hslash\omega_D])$.
\begin{lemma}
There is a constant $k$ \   $(0<k<1)$ satisfying
\[
\left\| Au-Av \right\| \leq k \left\| u-v \right\| \qquad \mbox{for all} \quad u,\, v \in V.
\]
\end{lemma}

\begin{proof} \   Let $u,\, v \in V$. Then
\begin{eqnarray*}
\left| Au(T,\, x)-Av(T,\, x) \right| &\leq& U_2 \int_{\varepsilon}^{\hslash\omega_D} \left| \frac{u(T,\, \xi)}{\,\sqrt{\,\xi^2+u(T,\, \xi)^2\,}\,}\,
\tanh \frac{\,\sqrt{\,\xi^2+u(T,\, \xi)^2\,}\,}{2T} \right. \\
& &\qquad \left. -\frac{v(T,\, \xi)}{\,\sqrt{\,\xi^2+v(T,\, \xi)^2\,}\,}\,
\tanh \frac{\,\sqrt{\,\xi^2+v(T,\, \xi)^2\,}\,}{2T} \right| \, d\xi.
\end{eqnarray*}
Note that each of $\displaystyle{ \tanh \frac{\,\sqrt{\,\xi^2+u(T,\, \xi)^2\,}\,}{2T} }$ and $\displaystyle{ \tanh \frac{\,\sqrt{\,\xi^2+v(T,\, \xi)^2\,}\,}{2T} }$ is regarded as $1$ when $T=0$. The integrand above becomes
\[
g\left( T;\, \xi,\, c(T,\,\xi,\, u,\, v)\right) \left| u(T,\, \xi)-v(T,\, \xi) \right|,
\]
where $g$ is that in \eqref{eq:functiong}. Here, $c(T,\,\xi,\, u,\, v)$ depends on $T$, $\xi$, $u$ and $v$, and satisfies $u(T,\, \xi)<c(T,\,\xi,\, u,\, v)<v(T,\, \xi)$ or $v(T,\, \xi)<c(T,\,\xi,\, u,\, v)<u(T,\, \xi)$.

By \eqref{eq:t1condition}, $\Delta_2^{-1}\left( \Delta_0(T_1) \right)$ satisfies \eqref{eq:tcondition}. Here, $T_2$ in \eqref{eq:tcondition} is replaced by $\Delta_2^{-1}\left( \Delta_0(T_1) \right)$ and $X$ in \eqref{eq:tcondition} is replaced by $c(T,\,\xi,\, u,\, v)$, respectively. Since
\[
\Delta_2^{-1}\left( \Delta_0(T) \right) \leq \Delta_2^{-1}\left( \Delta_0(T_1) \right),
\]
it follows from Lemma \ref{lm:increasing} that
\[
g\left( T;\, \xi,\, c(T,\,\xi,\, u,\, v)\right) \leq
g\left( \Delta_2^{-1}\left( \Delta_0(T) \right);\, \xi,\, c(T,\,\xi,\, u,\, v)\right).
\]
Note that \  $\displaystyle{ \frac{Z}{\,\cosh^2Z\,}<\tanh Z \quad (Z>0) }$ and that the function $\displaystyle{ Z \mapsto \frac{\,\tanh Z\,}{Z} }$ is strictly decreasing on $(0,\,\infty)$. Hence
\[
g\left( \Delta_2^{-1}\left( \Delta_0(T) \right);\, \xi,\, c(T,\,\xi,\, u,\, v)\right) \leq
\frac{1}{\,\sqrt{\,\xi^2+\Delta_1(T)^2\,}\,}\,
\tanh \frac{\,\sqrt{\,\xi^2+\Delta_1(T)^2\,}\,}{2\,\Delta_2^{-1}\left( \Delta_0(T) \right)}.
\]
Note again that the function
\[
T \mapsto \int_{\varepsilon}^{\hslash\omega_D}
\frac{U_2}{\,\sqrt{\,\xi^2+\Delta_1(T)^2\,}\,}\,
\tanh \frac{\,\sqrt{\,\xi^2+\Delta_1(T)^2\,}\,}{2\,\Delta_2^{-1}\left( \Delta_0(T) \right)}\, d\xi
\]
is continuous on $[0,\, T_1]$. Set
\[
k=\max_{T \in [0,\, T_1]} \int_{\varepsilon}^{\hslash\omega_D}
\frac{U_2}{\,\sqrt{\,\xi^2+\Delta_1(T)^2\,}\,}\,
\tanh \frac{\,\sqrt{\,\xi^2+\Delta_1(T)^2\,}\,}{2\,\Delta_2^{-1}\left( \Delta_0(T) \right)}\, d\xi.
\]
Then \  $\displaystyle{ \left\| Au-Av \right\| \leq k \left\| u-v \right\| }$. By \eqref{eq:smplgapequation2},
\begin{eqnarray*}
& & k \\
&\leq& \max_{T \in [0,\, T_1]} \int_{\varepsilon}^{\hslash\omega_D}
\frac{U_2}{\,\sqrt{\,\xi^2+\Delta_0(T)^2\,}\,}\,
\tanh \frac{\,\sqrt{\,\xi^2+\Delta_0(T)^2\,}\,}{2\,\Delta_2^{-1}\left( \Delta_0(T) \right)}\, d\xi \\
&=& \int_{\varepsilon}^{\hslash\omega_D}
\frac{U_2}{\,\sqrt{\,\xi^2+\left\{ \Delta_2\left( \Delta_2^{-1}\left( \Delta_0(T) \right) \right) \right\}^2\,}\,}\, \tanh
\frac{\,\sqrt{\,\xi^2+\left\{ \Delta_2\left( \Delta_2^{-1}\left( \Delta_0(T) \right) \right) \right\}^2
\,}\,}{2\,\Delta_2^{-1}\left( \Delta_0(T) \right)}\, d\xi\\
&=& 1.
\end{eqnarray*}
From the lemma just below, the result thus follows.
\end{proof}

\begin{lemma}\   Let $k$ be as above. Then $0<k<1$.
\end{lemma}

\begin{proof} \  Assume that $k=1$. Then there is a $\tau \in [0,\, T_1]$ satisfying
\[
1=k=\int_{\varepsilon}^{\hslash\omega_D}
\frac{U_2}{\,\sqrt{\,\xi^2+\Delta_1(\tau)^2\,}\,}\,
\tanh \frac{\,\sqrt{\,\xi^2+\Delta_1(\tau)^2\,}\,}{2\,\Delta_2^{-1}\left( \Delta_0(\tau) \right)}\, d\xi.
\]
Hence
\begin{equation}\label{eq:deltatwozero}
1=\int_{\varepsilon}^{\hslash\omega_D}
\frac{U_2}{\,\sqrt{\,\xi^2+\left\{ \Delta_2\left( \Delta_2^{-1}\left( \Delta_1(\tau) \right) \right) \right\}^2\,}\,}\,
\tanh \frac{\,\sqrt{\,\xi^2+\left\{ \Delta_2\left( \Delta_2^{-1}\left( \Delta_1(\tau) \right) \right) \right\}^2\,}\,}{2\,\Delta_2^{-1}\left( \Delta_0(\tau) \right)}\, d\xi.
\end{equation}
It follows form \eqref{eq:smplgapequation2} that
\begin{equation}\label{eq:deltatwoone}
1=\int_{\varepsilon}^{\hslash\omega_D}
\frac{U_2}{\,\sqrt{\,\xi^2+\left\{ \Delta_2\left( \Delta_2^{-1}\left( \Delta_1(\tau) \right) \right) \right\}^2\,}\,}\,
\tanh \frac{\,\sqrt{\,\xi^2+\left\{ \Delta_2\left( \Delta_2^{-1}\left( \Delta_1(\tau) \right) \right) \right\}^2\,}\,}{2\,\Delta_2^{-1}\left( \Delta_1(\tau) \right)}\, d\xi.
\end{equation}
Comparison of \eqref{eq:deltatwozero} and \eqref{eq:deltatwoone} gives $\Delta_2^{-1}\left( \Delta_0(\tau) \right)=\Delta_2^{-1}\left( \Delta_1(\tau) \right)$, and hence $\Delta_0(\tau)=\Delta_1(\tau)$. This contradicts Lemma \ref{lm:taudelta} (b).
\end{proof}

By the Banach fixed-point theorem (see e.g. Zeidler \cite[pp.18-22]{zeidler}), there is a unique $u_1 \in V$ satisfying $\displaystyle{ Au_1=u_1 }$. Let us fix $T \in [0,\, T_1]$. Then, for each $u \in V$, it follows that $u(T,\,\cdot) \in V_T$. Here, $V_T$ is that in \eqref{eq:vt}. Theorem \ref{thm:solutionu0} thus imples $u_1=u_0$. This completes the proof of Theorem \ref{thm:continuityu0}.

%%%%%%%%%%%%%%%%%%%%%%%%%%%%%%%%%%%%%%%%%%%%%%%%%%
\noindent \textbf{Acknowledgments}

S. Watanabe is supported in part by the JSPS Grant-in-Aid for Scientific Research (C) 21540110.

%%%%%%%%%%%%%%%%%%%%%%%%%%%%%%%%%%%%%%%%%%%%%%%%%%

\end{document}